\documentclass{article}

\usepackage{amsthm}
\usepackage{amsmath}
\usepackage{amssymb}

\usepackage[all]{xy}

\newtheorem{thm}{Theorem}

\def\To{\Rightarrow}
\def\TTo{\Longrightarrow}

\def\ie{i.e.}

\title{Matrix Insertion-Deletion Systems}

\author{Ion Petre$^{1}$\and Sergey Verlan$^2$}

\date{
$^1$Department of IT, \AA bo Akademi University, Turku 20520 Finland
\\
\texttt{ipetre@abo.fi}
\\
$^2$Laboratoire d'Algorithmique, Complexit\'e et Logique,\\
D\'epartement Informatique,\\
Universit\'e Paris Est,\\
 61, av. G\'en\'eral de Gaulle, 94010 Cr\'eteil,
France
\\
\texttt{verlan@univ-paris12.fr}
}

\newcommand{\prule}[3]{\ensuremath{(#1,#2,#3)}}
\newcommand{\ins}[3]{\ensuremath{\prule{#1}{#2}{#3}_{ins}}}
\newcommand{\del}[3]{\ensuremath{\prule{#1}{#2}{#3}_{del}}}
\newcommand{\cins}[2]{\ins{#1}{#2}{\lambda}}
\newcommand{\cdel}[2]{\del{#1}{#2}{\lambda}}
\newcommand{\fins}[1]{\ins{\lambda}{#1}{\lambda}}
\newcommand{\fdel}[1]{\del{\lambda}{#1}{\lambda}}
\newcommand{\mrule}[2]{\ensuremath{\left[#1,\ #2\right]}}
\newcommand{\mruleI}[3]{\mrule{\fins{#1}}{\cdel{#2}{#3}}}
\begin{document}

\maketitle

\begin{abstract}
In this article, we consider for the first time the operations of insertion and
deletion working in a matrix controlled manner. We show that, similarly as in
the case of context-free productions, the computational power is strictly
increased when using a matrix control: computational completeness can be
obtained by systems with insertion or deletion rules involving at most two
symbols in a contextual or in a context-free manner and using only binary
matrices.
\end{abstract}

\section{Introduction}

The operations of insertion and deletion were first considered with a
linguistic motivation~\cite{Marcus,Galiuk,Kluwer}. Another inspiration for
these operations comes from the fact that the insertion operation and its
iterated variants are generalized versions of Kleene's operations of
concatenation and closure~\cite{Kleene56}, while the deletion operation
generalizes the quotient operation. A study of properties of the corresponding
operations may be found in~\cite{Haussler82,Haussler83,Kari}. Insertion and
deletion also have interesting biological motivations, e.g., they correspond to
a mismatched annealing of DNA sequences; these operations are also present in
the evolution processes in the form of point mutations as well as in RNA
editing, see the discussions in~\cite{Beene,BBD07,Smith} and~\cite{dna}. These
biological motivations of insertion-deletion operations led to their study in
the framework of molecular computing, see, for example,
\cite{Daley,cross,dna,TY}.

In general, an insertion operation means adding a substring to a given string
in a specified (left and right) context, while a deletion operation means
removing a substring of a given string from a specified (left and right)
context. A finite set of insertion-deletion rules, together with a set of
axioms provide a language generating device: starting from the set of initial
strings and iterating insertion-deletion operations as defined by the given
rules, one obtains a language.

Even in their basic variants, insertion-deletion systems are able to
characterize the recursively enumerable languages. Moreover, as it was shown in
\cite{cfinsdel}, the context dependency may be replaced by insertion and
deletion of strings of sufficient length, in a context-free manner. If the
length is not sufficient (less or equal to two) then such systems are not able
to generate more than the recursive languages and a characterization of them
was shown in~\cite{SV2-2}.

Similar investigations were continued in \cite{MRV07,KRV08,KRV08c} on
insertion-deletion systems with one-sided contexts, i.e., where the context
dependency is present only from the left or only from the right side of all
insertion and deletion rules. The papers cited above give several computational
completeness results depending on the size of insertion and deletion rules. We
recall the interesting fact that some combinations are not leading to
computational completeness, i.e., there are recursively enumerable languages
that cannot be generated by such devices.

Like in the case of context-free rewriting, it is possible to consider a
graph-controlled variant of insertion-deletion systems. Thus the rules cannot
be applied at any time, as their applicability depends on the current
``state'', changed by a rule application. Such a formalization is rather
similar to the definition of insertion-deletion P systems~\cite{membr}, however
it is even simpler and more natural. The article~\cite{FKRV10}  focuses on
one-sided graph-controlled insertion-deletion systems where at most two symbols
may be present in the description of insertion and deletion rules. This
correspond to systems of size $(1,1,0;1,1,0)$, $(1,1,0;1,0,1)$,
$(1,1,0;2,0,0)$, and $(2,0,0;1,1,0)$, where the first three numbers represent
the maximal size of the inserted string and the maximal size of the left and
right contexts, while the last three numbers represent the same information,
but for deletion rules. It is known that such systems are not computationally
complete~\cite{KRV11}, while the corresponding P systems variants and
graph-controlled variants are computationally complete.

In this article we introduce a new type of control, similar to the one used in
matrix grammars. More precisely, insertion and deletion rules are grouped in
sequences, called matrices, and either the whole sequence is applied
consecutively, or no rule is applied. We show that in the case of such control
the computational power of systems of size $(1,1,0;2,0,0)$ and $(2,0,0;1,1,0$
is strictly increasing. Moreover, we show that binary matrices suffice to
achieve this result, hence we obtain a similar characterization like in the
case of the binary normal form for matrix grammars.

\section{Definitions}\label{sec:def}

We do not present the usual definitions concerning standard concepts of the
theory of formal languages and we only refer to textbooks such as
\cite{handbook} for more details.

The empty string is denoted by $\lambda $.
%For the interval of natural numbers
%from $k$ to $m$ we write $\left[ k..m\right] $.

In the following, we will use special variants of the \emph{Geffert} normal
form for type-0 grammars (see~\cite{Geffert91} for more details).

A grammar $G=\left( N,T,P,S\right) $ is said to be in \emph{Geffert normal
form}~\cite{Geffert91} if $N=\{S,A,B,C,D\}$ and $P$ only contains context-free
rules of the forms $S\to uSv$ with $u\in \{A,C\}^*$ and $v\in \{B,D\}^*$ as
well as $S\to x$ with $x\in (T\cup \{A,B,C,D\})^*$ and two (non-context-free)
erasing rules $AB\to \lambda $ and $CD\to \lambda $.

We remark that we can easily transform the linear rules from the Geffert normal
form into a set of left-linear and right-linear rules (by increasing the number
of non-terminal symbols, e.g., see \cite{membr}). More precisely, we say that a
grammar $G=\left( N,T,P,S\right) $ with $N=N'\cup N''$, $S,S'\in N'$, and
$N''=\{A,B,C,D\}$, is in the \emph{special Geffert normal form} if, besides the
two erasing rules $AB\to \lambda $ and $CD\to \lambda $, it only has
context-free rules of the following forms:

\begin{align*}
& X\to bY,\quad X,Y\in N',b\in T\cup N'', \\
& X\to Yb,\quad X,Y\in N',b\in T\cup N'', \\
& S'\to \lambda .
\end{align*}

Moreover, we may even assume that, except for the rules of the forms $%
X\to Sb$ and $X\to S'b$, for the first two types of
rules it holds that the right-hand side is unique, i.e., for any two rules $%
X\to w$ and $U\to w$ in $P$ we have $U=X$.

The computation in a grammar in the special Geffert normal form is done in two
stages. During the first stage, only context-free rules are applied. During the
second stage, only the erasing rules $AB\to \lambda $ and $CD\to \lambda $ are
applied. These two erasing rules are not applicable during the first stage as
long as the left and the right part of the current string are still separated
by $S$ (or $S'$) as all the symbols $A$ and $C$ are generated on the left side
of these middle symbols and the corresponding symbols $B$ and $D$ are generated
on the right side. The transition between stages is done by the rule $S'\to
\lambda $. We remark that all these features of a grammar in the special
Geffert normal form are immediate consequences of the proofs given
in~\cite{Geffert91}.

\subsection{Insertion-deletion systems}

An \textit{insertion-deletion system} is a construct $ID=(V,T,A,I,D)$, where
$V$ is an alphabet; $T\subseteq V$ is the set of \textit{terminal} symbols (in
contrast, those of $V-T$ are called \textit{non-terminal} symbols); $A$
is a finite language over $V$, the strings in $A$ are the \textit{axioms}; $%
I,D$ are finite sets of triples of the form $(u,\alpha ,v)$, where $u$, $%
\alpha $ ($\alpha \neq \lambda $), and $v$ are strings over $V$. The triples in
$I$ are \textit{insertion rules}, and those in $D$ are \textit{deletion
rules}. An insertion rule $(u,\alpha ,v)\in I$ indicates that the string $%
\alpha $ can be inserted between $u$ and $v$, while a deletion rule $%
(u,\alpha ,v)\in D$ indicates that $\alpha $ can be removed from between the
context $u$ and $v$. Stated in another way, $(u,\alpha ,v)\in I$ corresponds to
the rewriting rule $uv\to u\alpha v$, and $(u,\alpha ,v)\in D$ corresponds to
the rewriting rule $u\alpha v\to uv$. By $\To _{ins}$ we denote the relation
defined by the insertion
rules (formally, $x\To_{ins}y$ if and only if $%
x=x_{1}uvx_{2},y=x_{1}u\alpha vx_{2}$, for some $(u,\alpha ,v)\in I$ and
$x_{1},x_{2}\in V^*$), and by $\To_{del}$ the relation defined by the deletion
rules (formally, $x\To_{del}y$ if and only if $x=x_{1}u\alpha
vx_{2},y=x_{1}uvx_{2}$, for some $(u,\alpha ,v)\in D$ and $x_{1},x_{2}\in
V^*$). By $\To $ we refer to any of the relations $\To_{ins},\To_{del}$, and by
$\To^*$ we denote the reflexive and transitive closure of $\To$.

The language generated by $ID$ is defined by
\begin{equation*}
L(ID)=\{w\in T^*\mid x\To ^*w\mathrm{\ for\ some\ }
x\in A\}.
\end{equation*}

The complexity of an insertion-deletion system $ID=(V,T,A,I,D)$ is described by
the vector $(n,m,m';p,q,q')$ called \emph{size}, where \vspace{-2mm}
\begin{eqnarray*}
n=\max\{|\alpha|\mid (u,\alpha,v)\in I\}, & & p=\max\{|\alpha|\mid
(u,\alpha,v)\in D\}, \\
m=\max\{|u|\mid (u,\alpha,v)\in I\}, & & q=\max\{|u|\mid (u,\alpha,v)\in D\},
\\
m'=\max\{|v|\mid (u,\alpha,v)\in I\}, & & q'=\max\{|v|\mid
(u,\alpha,v)\in D\}.
\end{eqnarray*}

By $INS_{n}^{m,m'}DEL_{p}^{q,q'}$ we denote the families of insertion-deletion
systems having the size $(n,m,m';p,q,q')$.

If one of the parameters $n,m,m',p,q,q'$ is not specified, then instead we
write the symbol~$\ast $. In particular, $INS_*^{0,0}DEL_*^{0,0}$ denotes the
family of languages generated by \emph{context-free insertion-deletion
systems}. If one of numbers from the pairs $m $, $m'$ and/or $q$, $q'$ is equal
to zero (while the other one is not), then we say that the corresponding
families have a one-sided context. Finally we remark that the rules from $I$
and $D$ can be put together into one set of rules $R$ by writing $\left(
u,\alpha ,v\right) _{ins}$ for $\left( u,\alpha ,v\right) \in I$ and $\left(
u,\alpha ,v\right) _{del}$ for $\left( u,\alpha ,v\right) \in D$.

\subsection{Matrix insertion-deletion systems}

Like context-free grammars, insertion-deletion systems may be extended by
adding some additional controls. We discuss here the adaptation of the idea of
matrix grammars for insertion-deletion systems.

A \emph{matrix insertion-deletion system} is a construct%
\begin{equation*}
\gamma =(V,T,A,M)\mathrm{\ where}
\end{equation*}

\begin{itemize}
\item $V$ is a finite alphabet,

\item $T\subseteq V$ is the \emph{terminal alphabet},

\item $A\subseteq V^*$ is a finite set of \emph{axioms},

\item $M=r_1,\dots, r_n$ is a finite set of sequences of rules, called
    \emph{matrices}, of the form $r_i:[r_{i1},\dots{}r_{ik}]$ where
    $r_{ij}$, $1\le i\le n$, $1\le j\le k$ is an insertion or deletion rule
    over $V$.
\end{itemize}

The sentential form (also called configuration) of $\gamma$ is a string $w\in
V^*$. A transition $w\TTo_{r_i}w'$, for $1\le i\le n$, is performed if there
exist $w_1,\dots,w_k\in V^*$ such that
$w\To_{r_{i1}}w_1\To_{r_{i2}}\dots{}\To_{r_{ik}}w_k$ and $w_k=w'$.

The language generated by $\gamma$ is defined by
\begin{equation*}
L(\gamma)=\{w\in T^*\mid x\TTo ^*w\mathrm{\ for\ some\ }
x\in A\}.
\end{equation*}

By $Mat_kINS_{n}^{m,m'}DEL_{p}^{q,q'}$, $k>1$, we denote the families of matrix
insertion-deletion systems having matrices with at most $k$ rules and insertion
and deletion rules of size $(n,m,m';p,q,q')$.

\section{Computational completeness}

For all the variants of insertion and deletion rules considered in this
section, we know that the basic variants without using matrix control cannot
achieve computational completeness (see \cite{KRV11}, \cite{MRV07}). The
computational completeness results from this section are based on simulations
of derivations of a grammar in the special Geffert normal form. These
simulations associate a group of insertion and deletion rules to each of the
right- or left-linear rules $X\to bY$ and $X\to Yb$. The same holds for
(non-context-free) erasing rules $AB\to \lambda $ and $CD\to \lambda $. We
remark that during the derivation of a grammar in the special Geffert normal
form, any sentential form contains at most one non-terminal symbol from $N'$.

We start with the following affirmation: if the size of the matrices is
sufficiently large, then corresponding systems are computationally complete.
This is quite obvious for matrices of size 3.

\begin{thm}
$Mat_3INS_1^{1,0}DEL_2^{0,0} = RE.$
\end{thm}

\begin{proof}
The proof is based on a simulation of a type-0 grammar in the Geffert normal
form (as presented in Section~\ref{sec:def}). Let $G=(V,T,S,P)$ be such a
grammar. We construct the system $\gamma=(V,T,\{S\},M)$ as follows.

%For every rule $r:A\to bC\in P$ we add to $M$ the matrix\\
%$r:[\cins{A}{C},\cins{A}{b},\fdel{A}]$.

For every rule $r:A\to xy\in P$, $x,y\in V$ we add to $M$ the matrix\\
$r:[\cins{A}{y},\cins{A}{x},\fdel{A}]$.

%For every rule $r:A\to Cb\in P$ we add to $M$ the matrix\\
%$r:[\cins{A}{b},\cins{A}{C},\fdel{A}]$.

For rules $AB\to\lambda\in P$, $CD\to\lambda\in P$ and $S'\to\lambda$ we add to
$M$ following matrices:\\
$AB: [\fdel{AB},\fins{\lambda},\fins{\lambda}]$,\\
$CD:[\fdel{CD},\fins{\lambda},\fins{\lambda}]$ and\\
$S':[\fdel{S'},\fins{\lambda},\fins{\lambda}]$.

It is clear that $L(\gamma)=L(G)$. Indeed, rules of type $A\to xy$ are
simulated by consecutively inserting $y$ and $x$ after $A$ and finally deleting
$A$. The rules $AB\to\lambda$ and $CD\to\lambda$ are simulated by directly
erasing 2 symbols and the rule $S'\to\lambda$ by directly erasing $S'$.
\end{proof}

A similar result can be obtained in the case of systems having rules of size
$(1,1,0;1,1,0)$.

\begin{thm}
$Mat_3INS_1^{1,0}DEL_1^{1,0} = RE.$
\end{thm}

\begin{proof}
The proof is done like in the previous theorem. Right- and left-linear rules
are simulated exactly in the same manner. The rule $AB\to\lambda$ can be
simulated by three matrices (providing that the axiom is $\{\$S\}$):\\
$AB.1:[\fins{K_{AB}},\fdel{\$},\fins{\lambda}]$,\\
$AB.2:[\cdel{K_{AB}}{A},\cdel{K_{AB}}{B},\fins{\lambda}]$ and\\
$AB.3:[\fdel{K_{AB}},\fins{\$},\fins{\lambda}]$.\\

They simulate $AB\to\lambda$ by introducing in the string  a symbol $K_{AB}$ in
a context-free manner and after that by deleting one copy of adjacent $A$ and
$B$. The validity follows from the observation that there can be at most only
one copy of $K_{AB}$ in the string (because it's insertion and deletion is
synchronized with the deletion and insertion of a special symbol $\$$ initially
present in only one copy). In order to delete this symbol at the end of the
computation the matrix $[\fdel{\$},\fins{\lambda},\fins{\lambda}]$ shall be
used.

The rule $CD\to\lambda$ is simulated similarly and the rule $S'\to\lambda$ by
directly erasing $S'$.
\end{proof}

By taking deletion rules with a right context in the previous theorem we
obtain.

\begin{thm}
$Mat_3INS_1^{1,0}DEL_1^{0,1} = RE.$
\end{thm}

We give below the proof for the case of systems of size $(2,0,0;1,1,0)$.

\begin{thm}
$Mat_3INS_2^{0,0}DEL_1^{1,0} = RE.$
\end{thm}

\begin{proof}
The proof is based on a simulation of a type-0 grammar in the Geffert normal
form (as presented in Section~\ref{sec:def}). Let $G=(V,T,S,P)$ be such a
grammar. We construct the system $\gamma=(V\cup V',T,\{S\},M)$ as follows
($V'=\{X_A,Y_A\mid A\in V\}\cup \{K_{AB},K_{CD}\})$.

For every rule $r:A\to bC\in P$ we add to $M$ the matrix\\
$r:[\fins{bC},\cdel{C}{A}]$.

For every rule $r:A\to Cb\in P$ we add to $M$ the matrices\\
$r.1:[\fins{X_AY_A},\cdel{Y_A}{A}]$,\\
$r.2:[\fins{Cb},\cdel{b}{Y_A},\fdel{\lambda}]$ and\\
$r.3:[\fdel{\$},\fdel{X_A},\fins{\$}]$.

For rules $AB\to\lambda\in P$ and $CD\to\lambda\in P$ we add to $M$ following
six matrices:
 {\small
\begin{align*}
AB.1&:[\fins{K_{AB}},\fdel{\$},\fins{\lambda}], &CD.1&:[\fins{K_{CD}},\fdel{\$},\fins{\lambda}],\\
AB.2&:[\cdel{K_{AB}}{A},\cdel{K_{AB}}{B},\fins{\lambda}],& CD.2&:[\cdel{K_{CD}}{C},\cdel{K_{CD}}{D},\fins{\lambda}],\\
AB.3&:[\fdel{K_{AB}},\fins{\$},\fins{\lambda}], &CD.3&:[\fdel{K_{CD}},\fins{\$},\fins{\lambda}].
\end{align*}
}

The rule $S'\to\lambda$ is simulated by the matrix that introduces symbol $\$$
$S':[\fdel{S'},\fins{\$},\fins{\lambda}]$.

It is clear that $L(\gamma)=L(G)$. Indeed, any rule $A\to bC$ is simulated
directly by inserting $bC$ and deleting $A$ in the context of $C$. The right
position for the insertion is insured by the uniqueness of $A$. Rules $r:A\to
Cb$ are simulated in a different way. First $A$ is replaced by $X_AY_A$ and
after that $Y_A$ is rewritten by $Cb$ as in the previous case. We remark that
by inserting $X_AY_A$ we insure that there is no symbol $b$ before $Y_A$. This
permits to correctly place $Cb$. The additional symbol $X_A$ remaining in the
string is deleted during the second stage (when symbol $\$$ is introduced). As
before, in order to delete $\$$ at the end of the computation the matrix
$[\fdel{\$},\fins{\lambda},\fins{\lambda}]$ shall be used.
\end{proof}

Since the matrix control is a particular case of the graph control we obtain
\begin{thm}\cite{KRV11}
For any $k>0$, $REG\setminus Mat_kINS_2^{0,0}DEL_2^{2,0} \ne \emptyset.$
\end{thm}

\section{Computational completeness for binary matrices}

In this section we show that binary matrices suffice for computational
completeness.

\begin{thm}
$Mat_2INS_2^{0,0}DEL_1^{1,0} = RE.$
\end{thm}

\begin{proof}

The proof is based on a simulation if type-0 grammar in the Geffert normal form
(as presented in Section~\ref{sec:def}). Let $G=(V,T,S,P)$ be such a grammar.
We construct the system $\gamma=(V\cup V',T,w,M)$ as follows.

$V'=\{\#_k^r, K^r\mid r\in P, 1\le k\le 5\}\cup\{K_{AB},K_{CD},\$\}$, and
$w=\{\$S\}$.

%The system:

For every rule $r:A\to bC\in P$ we add to $M$ the matrix
$$r.1:\mruleI{bC}{C}{A}.$$

For every rule $r:A\to Cb\in P$ we add to $M$ following matrices:

\begin{align*}
&r.1: \mruleI{\#_1^r\#_2^r}{\#_2^r}{A}\\
&r.2: \mruleI{C}{C}{\#_1^r}\\
&r.3: \mruleI{\#_3^r\#_4^r}{\#_4^r}{\#_2^r}\\
&r.4: \mruleI{\#_5^rb}{b}{\#_4^r}\\
&r.5: \mrule{\fdel{\$}}{\fins{K^r}}\\
&r.6: \mrule{\fdel{K^r}}{\fins{\$}}\\
&r.7: \mrule{\cdel{K^r}{\#_3^r}}{\cdel{K^r}{\#_5^r}}\\
\end{align*}

For rules $AB\to\lambda\in P$ and $CD\to\lambda\in P$ we add to $M$ following
matrices:
\begin{align*}
&AB.1: \mrule{\fdel{\$}}{\fins{K_{AB}}} && AB.1: \mrule{\fdel{\$}}{\fins{K_{CD}}}\\
&AB.2: \mrule{\fdel{K_{AB}}}{\fins{\$}} && AB.2: \mrule{\fdel{K_{CD}}}{\fins{\$}}\\
&AB.3: \mrule{\cdel{K_{AB}}{A}}{\cdel{K_{AB}}{B}} && AB.3: \mrule{\cdel{K_{CD}}{C}}{\cdel{K_{CD}}{D}}\\
\end{align*}

The rule $S'\to\lambda$ can be simulated by the following matrix:
$$
S':\mrule{\fdel{S'}}{\fins{\lambda}}.
$$

We claim that $L(\gamma)=L(G)$. First we show that $L(\gamma)\supseteq L(G)$.
Let $w_1Aw_2$ be a sequential form in $G$ (initially $S$) and let $w_1Aw_2\To_r
w_1bCw_2$ be a derivation in $G$. We show that in $\gamma$ we obtain the same
result:
\begin{equation*}
w_1Aw_2\TTo_{r.1}  w_1bCw_2.
\end{equation*}

We remark that if the sequence $bC$ is not inserted before $A$, then the second
rule from the matrix will not be applicable (we recall that $w_1w_2$ does not
contain non-terminals from $V\setminus T$ and that $b\ne\lambda$.

Consider now the following derivation in $G$: $w_1Aw_2\To_r w_1Cbw_2$. This
derivation is simulated in $\gamma$ as follows.

\begin{multline*}
\$w_1Aw_2\TTo_{r.1} \$w_1\#_1^r\#_2^rw_2\TTo_{r.2} \$w_1C\#_2^rw_2\TTo_{r.3}\\
 \TTo_{r.3}\$w_1C\#_3^r\#_4^rw_2\TTo_{r.4}\$w_1C\#_3^r\#_5^rbw_2\TTo_{r.5}\\
 \TTo_{r.5}w_1CK^r\#_3^r\#_5^rbw_2
 \TTo_{r.6}
 w_1CK^rbw_2\TTo^{r.7}\$w_1Cbw_2.
\end{multline*}

So the grammar $G$ is simulated as follows. Firstly the first stage of the
generation is simulated by simulating left-linear and after that right-linear
productions. After changing to the second stage, rules $AB.i$ and $CD.i$ ($1\le
i\le 3$) can be applied, removing symbols $A,B,C,D$.

Now in order to prove the converse inclusion $L(\gamma)\subseteq L(G)$ we show
that no other words can be obtained in $\gamma$.

We start by observing  that matrices $r.1-r.4$ for a rule $r:A\to Cb$ as well
as $r.1$ for a rule $r:A\to bC$ have the form \mruleI{x}{x}{y}, $x\in V\cup
V^2,y\in V$. It is not difficult to see that if $x$ was not already present in
the string then such a matrix correspond to the rewriting rule $y\to x$.
Indeed, since $x$ is not present in the string, it should have been inserted
before the symbol $y$, which is deleted afterwards.

The matrices $r.5-r.7$ insure that a sequence of symbols $\#_3^r\#_5^r$ is
deleted. This is performed by introducing into the string a new special symbol
$K^r$. If it is not introduced before $\#_3^r$, then nothing happens and $K^r$
can be replaced by \$. Otherwise, it can delete the two symbols in the
sequence. The validity of the simulation is ensured by the fact that the symbol
\$ is always present in at most one copy.

In a similar way rules $AB.1-AB.3$ and $CD.1-CD.3$ act.

In order to conclude that the simulation of the rule $A\to Cb$ does not yield
other words we give the following remarks:
\begin{itemize}
\item Symbol $A$ is replaced by a pair of symbols $\#_1^r\#_2^r$, where
    $\#_1^r$ evolves to $C$ and $\#_2^r$ evolves to $b$.
\item Symbol $b$ is inserted if and only if $\#_3^r$ (and $\#_4^r$) is
    present in the string. This insures that this symbol is separated from
    any non-terminal that can be derived from $C$ and hence the insertion
    of this symbol cannot interfere with some other insertion that could be
    operated.
\end{itemize}

Since no other words can be generated we can reconstruct a derivation in $G$
starting from a derivation in $\gamma$. For this it is enough to follow
configurations where there is a non-terminal from $V\setminus \{A,B,C,D\}$ in
order to reconstruct the first stage of the derivation from $G$. The deletion
of $AB$ and $CD$ has a direct correspondence to the second stage of $G$. So,
$L(\gamma)=L(G)$.
\end{proof}

\begin{thm}
$Mat_2INS_1^{1,0}DEL_2^{0,0} = RE.$
\end{thm}

\begin{proof}

The proof is based on a simulation if type-0 grammar in the Geffert normal form
(as presented in Section~\ref{sec:def}). Let $G=(V,T,S,P)$ be such a grammar.
We construct the system $\gamma=(V\cup V',T,w,M)$ as follows.

$V'=\{p,p'\mid p: A\to Cb\in P\}\cup\{p,p_2,p_3,\#_p,\#_p',C_1^p,C_2^p\mid
p:A\to bC\in P\}\cup\{X,Y\}$, and $w=\{XSY\}$.

%The system:
%
%Initially $X\mathcal{S}Y$. $X$ and $Y$ are the positions for insertion and
%stack respectively.

For every rule $p:A\to Cb\in P$ we add to $M$ following matrices:
\begin{align*}
& p.1: \mrule{\fdel{A}}{\cins{Y}{p}}\\
& p.2: \mrule{\cins{X}{b}}{\cins{Y}{p'}}\\
& p.3: \mrule{\cins{X}{C}}{\fdel{p'p}}\\
\end{align*}

For every rule $p:A\to bC\in P$ we add to $M$ following matrices:
\begin{align*}
&p.1: \mrule{\fdel{A}}{\cins{Y}{p}}\\
&p.2: \mrule{\cins{X}{C_1^p}}{\cins{Y}{\#_p}}\\
&p.3: \mrule{\cins{Y}{\#_p'}}{\cins{Y}{p_2}}\\
&p.4: \mrule{\cins{Y}{p_3}}{\fdel{\#_p'\#_p}}\\
&p.5: \mrule{\cins{X}{b}}{\fdel{p_2}}\\
&p.6: \mrule{\fdel{X}}{\cins{C_1^p}{C_2^p}}\\
&p.7: \mrule{\cins{C_2^p}{X}}{\fdel{p_3p}}\\
&p.8: \mrule{\cins{X}{C}}{\fdel{C_1^pC_2^p}}\\
\end{align*}

For rules $AB\to\lambda\in P$ and $CD\to\lambda\in P$ we add to $M$ following
matrices:
\begin{align*}
&AB: \mrule{\fdel{AB}}{\fins{\lambda}} && CD: \mrule{\fdel{CD}}{\fins{\lambda}}
\end{align*}

We also add to $M$ the matrices $XY: \mrule{\fdel{X}}{\fdel{Y}}$ and
$S':\mrule{\fdel{S'}}{\fins{\lambda}}$.

We claim that $L(\gamma)=L(G)$. First we show that $L(\gamma)\supseteq L(G)$.
The simulation uses the following idea. Symbol $X$ marks the site where the
non-terminal is situated, while symbol $Y$ marks a position in the string (for
commodity we mark the end of the string). The sequence of insertions and
deletions is synchronized between these two positions: inserting something at
position $X$ also inserts or deletes symbols at position $Y$. Finally, symbols
at position $Y$ are checked to form some particular order. So in some sense $Y$
corresponds to a ``stack'' where some information is stored and after that the
``stack'' is checked to be in some specific form.

More precisely, let $w_1Aw_2$ be a sequential form in $G$ (initially $S$) and
let $w_1Aw_2\To_r w_1Cbw_2$ be a derivation in $G$. We show that in $\gamma$ we
obtain the same result:

\begin{multline*}
w_1XAw_2Y\TTo_{p.1} w_1Xw_2Yp\TTo_{p.2} w_1Xbw_2Yp'p\TTo_{p.2}^{k-1}\\
 \TTo_{p.2}^{k-1}w_1Xb^kw_2Yp'p^k\TTo_{p.3}w_1XCb^kw_2Yp'^{k-1}.
\end{multline*}

Since there are no rules eliminating $p'$ by itself (it can be eliminated only
if $p$ is following it, which is no more possible), the above string can become
terminal if and only if one insertion is done at the second step (i.e. $k=1$).
Hence we obtain the string $w_1XCbw_2Y$, \ie{} we correctly simulated the
corresponding production of the grammar.

We remark that the rule $p.2$ can be used at any time, but this yields again a
symbol $p'$ after $Y$ which cannot be removed.

Now consider the following derivation in $G$: $w_1Aw_2\To_r w_1bCw_2$. This
derivation is simulated in $\gamma$ as follows.

%\begin{multline*}
%w_1XAw_2Y\TTo_{p.1} w_1Xw_2Yp\TTo_{p.2} w_1XC_1^pw_2Y\#_pp\TTo_{p.2}^{k-1}\\
% \TTo_{p.2}^{k-1}w_1XC_1^{p,k}w_2Y\#_p^kp\TTo_{p.3}^m w_1XC_1^{p,k}w_2Y(p_2\#_p')^m\#_p^kp\TTo_{p.4}\\
% \TTo_{p.4} w_1XC_1^{p,k}w_2Yp_3(p_2\#_p')^{m-1}p_2\#_p^{k-1}p \TTo_{p.5}^n\\
% \TTo_{p.5}^n w_1Xb^sC_1^{p,k}w_2Yp_3(p_2\#_p')^{m-1-n-1}\#_p'^{n-1}\#_p^{k-1}p\TTo_{p.4}^s\\
% \TTo_{p.4}^s
% w_1Xb^sC_1^{p,k}w_2Yp_3^s(p_2\#_p')^{m-1-s-1}\#_p'^{n-1-s}\#_p^{k-1-s}p\TTo_{p.6}\\
% \TTo_{p.6}
% w_1b^sC_1^{p,k}C_2^pw_2Yp_3^s(p_2\#_p')^{m-1-s-1}\#_p'^{n-1-s}\#_p^{k-1-s}p\TTo_{p.7}\\
%\TTo_{p.7} w_1b^sC_1^pC_2^pXw_2Yp_3^{s-1}\TTo_{p.8} w_1b^sXCw_2Yp_3^{s-1}
%\end{multline*}
%
%We observe that if $s\ne 1$, then symbol $p_3$ cannot be eliminated, so we
%obtain the string $w_1XCbw_2Y$, \ie{} we correctly simulated the corresponding
%production of the grammar.

\begin{multline*}
w_1XAw_2Y\TTo_{p.1} w_1Xw_2Yp\TTo_{p.2} w_1XC_1^pw_2Y\#_pp\TTo_{p.3}\\
 \TTo_{p.3} w_1XC_1^{p}w_2Yp_2\#_p'\#_pp
 \TTo_{p.4} w_1XC_1^{p}w_2Yp_3p_2p \TTo_{p.5}\\
 \TTo_{p.5} w_1XbC_1^{p}w_2Yp_3p\TTo_{p.6} w_1bC_1^{p}C_2^pw_2Yp_3p\TTo_{p.7}\\
\TTo_{p.7} w_1bC_1^pC_2^pXw_2Y\TTo_{p.8} w_1b^sXCw_2Y
\end{multline*}

The deletion rules $AB\to\lambda$,  $CD\to\lambda$ and $S'\to\lambda$ are
simulated directly by rules $AB$ and $CD$ and symbols $X$ and $Y$ are
eliminated by the rule $XY$. Hence $L(G)\subseteq L(\gamma)$.

Now in order to prove the inclusion $L(\gamma)\subseteq L(G)$ we
 we show that only specific sequences of rule application
 %(in particular the
%sequence $p.1-p.8$)
can lead to a terminal string. The case of the simulation of rules of type
$A\to Cb$ is discussed above. We shall concentrate now on the simulation of
rules of type $A\to bC$. We give below the rules' dependency graph, where by
$x\leftarrow y$ we indicate that in order to apply $y$, we should apply at
least one time $x$.

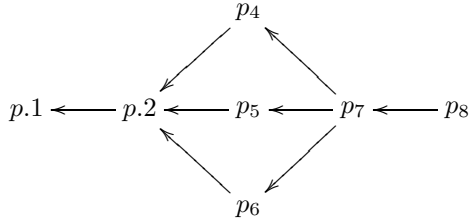
$$\xymatrix{
    &     &\ar[dl] p_4 & & \\
p.1 &\ar[l] p.2 &\ar[l] p_5 &\ar[l]\ar[ul]\ar[dl] p_7 &\ar[l] p_8\\
    &     &\ar[ul] p_6 & & \\
}
$$

Indeed, if rule $p.1$ is not applied first, then additional symbols are added
after $Y$ and it is clear that they cannot be eliminated. If $p.3$ is applied
before $p_2$, then the introduced symbol $\#_p'$ can never be deleted (as there
is no symbol $\#_p$ afterwards). Rule $p.4$ involves symbols introduced by
$p.2$ and $p.3$, so it cannot be used before. Rule $p.5$ cannot be applied
before $p.3$, however its application can be interchanged with the application
of $p.4$. Rule $p.6$ can be applied once after $p.2$, while in order to apply
$p.7$ we need to apply before rules $p.6$, $p.5$ and $p.4$. The rule $p.8$ is
applicable only after rule $r.7$.

It is still possible to apply some rules several times. Now we show that each
rule must be applied exactly once. Since there is only one copy of $A$, only
one copy of $p$ will be available. Hence rule $p.7$ will be applied only one
time. We can also deduce that only one copy of $p_3$ shall be produced, hence
rule $p.4$ should be applied only one time. But this implies the uniqueness of
symbols $\#_p$ and $\#_p'$, hence a single application of rules $p.2$ and
$p.3$. The last affirmation implies that $p_2$ is generated only once, hence
$p.5$ can be applied only once. From $p.6$ we can deduce that $C_2^p$ is
inserted once, hence $p.8$ is executed only one time. Finally, from the $p.2$
we can deduce that $C_1^r$ is inserted only once, so after its deletion in
$p.8$ no more copies will remain.

So, we obtain that any terminal derivation in $\gamma$ needs an application of
a specific sequence of rules. Hence, it is enough to look at strings from
$\gamma$ containing a non-terminal from $V\setminus \{A,B,C,D\}$ in order to
reconstruct the first stage of the derivation from $G$. The deletion of $AB$
and $CD$ has a direct correspondence to the second stage of $G$. This implies
that $L(\gamma)\subseteq L(G)$.
\end{proof}

%\begin{thm}
%$MatINS_1^{1,0}DEL_1^{1,0} = RE.$
%\end{thm}

\section{Conclusions}
In this article we have introduced the mechanism of a matrix control to the
operations of insertion and deletion.  We investigated the case of systems with
insertion and deletion rules of size $(1,1,0;1,1,0)$, $(1,1,0;1,0,1)$,
$(1,1,0;2,0,0)$ and $(2,0,0;1,1,0)$ and we have shown that the corresponding
matrix insertion-deletion systems are computationally complete. In the case of
first two systems matrices of size 3 are used, while in the case of the last
two systems binary matrices are sufficient. Since a matrix control is a
particular case of a graph control (having an input/output node and series of
linear paths starting and ending in this node), we obtain~\cite{KRV11} that
matrix insertion-deletion systems having rules of size $(2,0,0;2,0,0)$ are not
computationally complete.

We remark that our results for matrix insertion-deletion systems are different
from the results on graph-controlled systems obtained in~\cite{FKRV10} and
previous works. In the graph-controlled case, the total number of nodes in the
graph is minimized, while in the matrix case the depth of the graph
(corresponding to the size of matrices) is minimized.

We did not succeed to show the computational completeness of systems of size
$(1,1,0;1,1,0)$ and $(1,1,0;1,0,1)$ having binary matrices. This gives an
interesting topic for the further research.

\bibliographystyle{abbrv}
\bibliography{matinsdel}

\end{document}